\documentclass[letterpaper, 10pt, twocolumn]{ieeetran}

\IEEEoverridecommandlockouts                              

\usepackage{amsmath,amssymb}
\usepackage{euscript,yfonts,psfrag,latexsym,dsfont,graphicx,bbm,color}
\usepackage{amstext,wasysym,pdfsync}
\usepackage{epstopdf, bm}
\usepackage{amsmath, amssymb, psfrag, bm, latexsym, color, amstext}
\usepackage{graphicx}
\usepackage{epstopdf}
\usepackage{mathtools}
\usepackage{amsthm}
\usepackage{verbatim}
\usepackage{subcaption}
\usepackage[normalem]{ulem}

\usepackage{balance}
\usepackage{enumerate}

\usepackage{url}
\usepackage{float}

\usepackage{algorithm}
\usepackage{algorithmic}
\usepackage[normalem]{ulem}

\DeclareMathOperator*{\argmin}{argmin}

\DeclareMathOperator{\rank}{rank}
\DeclareMathOperator{\Span}{span}
\DeclareMathOperator{\dist}{dist}

\newtheorem{thm}{Theorem}

\newtheorem{lemma}[thm]{Lemma}
\newtheorem{prop}[thm]{Proposition}
\newtheorem{problem}{Problem}
\newtheorem{remark}{Remark}
\newtheorem{definition}[thm]{Definition}
\newtheorem{example}[thm]{Example}
\newtheorem{assumption}{Assumption}

\DeclarePairedDelimiter\ceil{\lceil}{\rceil}
\DeclarePairedDelimiter\floor{\lfloor}{\rfloor}

\graphicspath{{figures/}{../}}
\usepackage{xcolor}
\usepackage{etoolbox}

\newcommand{\tot}[1]{{\color{red}{#1}}}

\newcommand{\textrevisionPrev}[2]{} 

\title{\LARGE \bf
Network Consensus with Privacy: A Secret Sharing Method
}

\author{Silun Zhang,~\IEEEmembership{Member,~IEEE,}
        Thomas Ohlson Timoudas, 
        and~Munther Dahleh,~\IEEEmembership{Fellow,~IEEE}
\thanks{*This work was supported by the Knut and Alice Wallenberg Foundation, KAW 2018.0412. }
\thanks{$^{1}$ Department of Electrical Engineering and Computer
Science and Laboratory for Information and Decision Systems (LIDS), Massachusetts Institute of Technology, Cambridge, MA 02139 USA. ({\tt\small silunz@mit.edu})}
\thanks{$^{2}$ Division of network and systems engineering, KTH Royal Institute of Technology, Sweden.
({\tt\small ttohlson@kth.se})}
\thanks{$^{3}$ Department of Electrical Engineering and Computer
Science and Institute for Data, Systems, and Society, Massachusetts Institute of Technology, Cambridge, MA 02139 USA.
({\tt\small dahleh@mit.edu})}%
}

\begin{document}

\maketitle
\thispagestyle{empty}
\pagestyle{empty}

\begin{abstract}
In this work, inspired by secret sharing schemes, we introduce a privacy-preserving approach for network consensus, by which all nodes in a network can reach  an agreement on their states without exposing the individual state to neighbors.
With the \textit{privacy degree} defined for the agents, the proposed method makes the network resistant to the collusion of any given number of neighbors, and protects the consensus procedure from communication eavesdropping.
Unlike existing works, the proposed privacy-preserving algorithm is resilient to node failures. When a node fails, the  method offers the possibility of rebuilding the lost node via the information kept in its neighbors, even though none of the neighbors knows the exact state of the failing node. 
Moreover, it is shown that the proposed method can achieve consensus and average consensus almost surely, when the agents have arbitrary privacy degrees and a common privacy degree, respectively. To illustrate the theory, two numerical examples are presented. 






\end{abstract}


\section{Introduction}

Several key information and communication technologies are now converging. Internet of things (IoT), at the confluence of these, is already changing the ways we gather data, process this data, and ultimately derive value from it. 
The unprecedented scale of internet of things is forcing a shift towards increased decentralization, autonomy, and distributed computing \cite{ashton2009internet,chen2014iotvision,shi2016edgevision,wu2014CIoT}.


\textit{Privacy} and \textit{data ownership} - two hotly debated topics, and subject to several regulations in the recent years, e.g., GDPR - are major concerns in many internet of things applications \cite{yang2017iotprivacy,zanella2014iotsmartcities}. When sensitive data is exchanged between parties, there are no guarantees that it stays private, or that it isn't used for nefarious purposes. For instance, smart meter data could be used to infer the presence, absence, and even specific activities, of the  occupants in a house \cite{hart1992nonintrusive,shi2016edgevision}. Another example is location and navigation services, which have become indispensable for today's driving in increasingly intricate transport networks.
These location-aware services allow their providers to track the customers' movements, traveling intentions, and also living habits from the real-time location information shared with them \cite{beresford2003location}.
In recent years, many new methods and algorithms have been developed to limit the need to share data \cite{mcmahan2017federatedlearning}, or to preserve the sensitive components when data is shared in a network \cite{ambrosin2017odin,cortes2016differential,liu2019dynamical}. 

\textit{Reliability} is another important challenge in many IoT applications. The IoT networks are often vulnerable to node and/or communication failures due to various reasons, such as battery exhaustion, extreme environmental conditions, unprecedented radio interference, or even malicious attacks \cite{de2007survey}. 
Sometimes, intermittent failures are even to be expected, especially in applications involving low-energy devices in multi-hop networks, or employing network protocols developed for scalability and performance, rather than reliability \cite{korkmaz2010characterizing,sisinni2018industrial}. 
In particular, recovery and data reconstruction are big challenges with such fragile network architectures \cite{abbasi2012recovering,ang2017bigsensordata,di2016recovery,majdandzic2014spontaneous,shahriar2017generalized}. Protocols need to be aware of such situations, and be able to deal with them \cite{chiang2016iotoverview,stankovic2014iotdirections,weiner2014design}.

The aforementioned changes towards a decentralized network and computation architecture, have brought with them renewed interest in consensus problems, which often arise in decentralized decision making and computation. While consensus problems have a long history, dating back to at least the 1960's and 1970's (see, e.g., \cite{degroot1974reaching}), they have perhaps never been as relevant as they are today. 
The basic problem - namely, to device a method that allows a group of individuals or agents to reach agreement about some parameters through a decentralized communication protocol - naturally appears broadly in numerous IoT contexts, such as sensor fusion \cite{li2015fusionsurvey,olfati2005fusion,choi2012consensusWSN,yu2009consensusfiltering}, 
load balancing \cite{amelina2015loadbalancing,wang2011loadbalancing}, 
clock synchronization \cite{he2014timesynch,schenato2011clocksynch}, peak power load shedding \cite{xu2011loadshedding} and resource management \cite{zhao2018resourcemicrogrid} in smart grids,  distributed optimization \cite{nedic2009distributed,qu2017harnessing}, 
swarm coordination \cite{song2017intrinsic,zhang2018intrinsic,zhang2020intrinsic,zhang2020modeling},
distributed and federated learning \cite{li2020federatedlearning}, and large scale peer-to-peer networks \cite{jelasity04p2paggregate,kempe03aggregate}.

Several methods have been developed to achieve, at least to some extent, privacy-preserving consensus algorithms. Methods include masking the true state by
adding deterministic offsets to the messages \cite{altafini2019dynamical,gupta2017privacy,manitara2013privacy,rezazadeh2018privacy,wang2019statedecomposition,zhang2020consensus}, 
adding random noise to the messages transmitted amongst nodes \cite{he2018privacy,kefayati2007secure,huang2012differentially,mo2016privacy,manitara2013privacy},
and using various encryption schemes \cite{alexandru2019encrypted,ambrosin2017odin,ruan2019secure}. 
Another interesting method for computing separable functions without disclosing nodes' privacy appeared in \cite{ayaso2010information}, where agents exchange a set of samples drawn from a distribution depending on their true state, and the number of these samples can be tuned by a trade-off between the accuracy and privacy level of the algorithm.
While all of these methods have their own merits and drawbacks, none of them address the issue of data recovery. In this context, we still wish to mention \cite{bonawitz2017privacyaggregation}, which considers the case in which a central server collects data from all participating nodes, and computes their average, in a privacy-preserving way. They introduce a mechanism for recovering data from dropped nodes, but due to the presence of a central coordinator, both the task of computing the average, and the task of handling dropped nodes, are completely different from the decentralized situation.

In this paper, we present a network consensus algorithm which preserves individuals' privacy, using Shamir's secret sharing scheme \cite{beimel2011secret,shamir1979share}. 
In our proposed method, rather than sending its true local state, each agent only sends a secret share, generated from its true state, to its neighbors.
Moreover, these secret shares have the property that any less than a threshold number of them cannot reveal any information about the true state, but a set of more shares can  reconstruct the true state exactly.
This feature permits the proposed algorithm to preserve the agent's privacy  even with a collusion involving a certain number of  neighbors, and also protects the network from eavesdropping by an adversary.
In turn, due to the reconstructability of the true state, the overall information injected into the network is still sufficient to  reach network consensus but in a confidential way.  

Furthermore, we introduce the \textit{privacy degree} for each agent, which is the threshold number of colluding neighbors that can be allowed without privacy breach. We show that when all the agents share a common privacy degree, the network can reach an average consensus (see Theorem~\ref{thm:common_degree}), i.e., all agents' states eventually reach an agreement on the average of their initial states. 
On the contrary, in the scenario where different privacy degrees are employed for the agents, in order to guarantee the security and reconstructability of the true states, only consensus instead of average consensus can be obtained (see Theorem~\ref{thm:dif_degree}).

Compared to existing works, a striking merit of the proposed approach is that it is resilient to node failures, such as connection loss, memory loss and even permanent damage. If at any moment a node fails, the proposed method offers the possibility of rebuilding the lost node via the information kept in its neighbors, even though none of the neighbors knows the exact state of the failing node.
Moreover, unlike methods based on random perturbations or \textit{differential privacy} techniques, such as \cite{nozari2017differentially,huang2012differentially,kefayati2007secure}, the proposed method can reach average consensus with no errors when agents have a common privacy degree. 
In addition, the privacy security adopted in this paper  renders the network immune to the collusion of any given number of neighbors. 
In contrast to the existing deterministic methods, e.g., \cite{gupta2017privacy,manitara2013privacy,rezazadeh2018privacy},  where in order to reach consensus, the sum (integral) of all the added disturbance over time needs to be zero, and thus at least one neighbor of each agent must be honest.



%
%

The rest of this paper is organized as follows. Section~\ref{sec:pre} revisits some preliminary knowledge on  stochastic discrete-time systems and on secret sharing schemes. Section~\ref{sec:prob_formulation} gives the detailed definition of the privacy security, and the privacy-preserving consensus problem  that will be investigated in the work. 
Section~\ref{sec:alg_n_results} presents the  algorithms and the main theorems that this paper proposes.
In Section~\ref{sec:converge}, we analyze the convergence of the proposed algorithms and prove the main theorems  in the scenario of a common privacy degree and different privacy degrees, respectively. 
Then, Section~\ref{sec:simu} presents two numerical examples to illustrate the implementation of the proposed algorithms. Finally, we conclude in Section~\ref{sec:conclusion}.

\textbf{Notations:} By the symbols $\mathbb Z$, $\mathbb N$, and $\mathbb Z^+$, we denote the sets of all integers, nonnegative intergers, and positive integers, respectively. Given any positive integer $M \in \mathbb Z^+$, set $[M]=\{1, 2, \dots, M\}$. In addition, we denote by $|S|$ the cardinality of a given set $S$. 
For any event $\kappa$, define the indicator function $\mathbbm{1}(\cdot)$ satisfying that $\mathbbm{1}(\kappa) =1 $ when $\kappa$ happens, otherwise $\mathbbm{1}(\kappa) =0 $. Moreover, we denote by $\mathbf{1}_n \in \mathbb R^n$ and $\mathbf{1}_{m\times n} \in \mathbb R^{m\times n}$  the vector and matrix composed of all one entries, respectively.
The symbol $I_n$ is reserved for the identity matrix  of dimension $n$. 
For any square matrix $A$, $\sigma(A)$ is the set consisting of all eigenvalues of $A$.


\section{Preliminary}\label{sec:pre}
In this paper, we model the inter-agent connectivity in a networked system by a graph $\mathcal{G}=(\mathcal{V},\mathcal{E})$, where the set of nodes is $\mathcal{V}=\big\{1,\ldots,N\big\}$, and set $\mathcal{E}\subset \mathcal{V} \times \mathcal{V}$ consisting of all edges in the graph. 
A graph $\mathcal G$ is called undirected if $(i,j)\in \mathcal E$, for any $(j,i) \in \mathcal E$. 
In the rest of this paper, without further indication, we assume all  graphs are undirected.
With the connectivity graph $\mathcal G$, agents $i$ and $j$ can send information to each other if and only if $(i,j) \in \mathcal E$.
Then the neighbor set of a node $i$ is defined as $\mathcal{N}_i= \big\{ j:(i,j) \in \mathcal{E}\big\}$, and a node $j$ is called a neighbor of $i$, if $j \in \mathcal{N}_i$. 
Moreover, we say an edge is \textit{incident} to a vertex $i\in \mathcal V$ if $i$ is one of the endpoints of the edge, and two edges are \textit{adjacent}  if they are incident to a common node. 

\textrevisionPrev{
The communication between $N$ agents will be modelled by an undirected\footnote{Recall that a graph $\mathcal G$ is called undirected if $(i,j)\in \mathcal E$, whenever $(j,i) \in \mathcal E$.} graph $\mathcal{G}=(\mathcal{V},\mathcal{E})$, with the set of nodes $\mathcal{V}=\big\{1,\ldots,N\big\}$ corresponding to the set of agents, and the edges $\mathcal{E}\subset \mathcal{V} \times \mathcal{V}$ corresponding to the communication links between them. \tot{**DO WE EVER USE DIRECTED GRAPHS? IF NOT, REMOVE THIS** In the rest of this paper, without further indication, we assume all the graphs are undirected.}
In this model, an agent $i$ can therefore send a message to agent $j$ if and only if there is an edge between them (i.e., $(i, j) \in \mathcal{E}$), and it is furthermore implied that any message sent over a given communication link (edge) is only received by the other endpoint (unless another agent is eavesdropping).

The neighbor set of a node $i$ is defined as $\mathcal{N}_i= \big\{ j:(i,j) \in \mathcal{E}\big\}$, and a node $j$ is called a neighbor of $i$, if $j \in \mathcal{N}_i$. Moreover, an edge is said to be \textit{incident} to a node $i\in \mathcal V$ if $i$ is one of the endpoints of the edge, and two edges are called \textit{adjacent} if they are incident to a common node. 
}





%
%
%
%

\subsection{Invariance principle for stochastic systems}
We first revisit some results of the invariance principle for discrete-time stochastic systems.

Let $\mathcal S$ be a subspace of $\mathbb R^n$. Given a complete probability space $(\Omega, \mathcal F, \mathbb P)$, consider the stochastic system
\begin{align}\label{eq:stoch_system}
x(t+1)= F(x(t), w_t),
\end{align}
with $x(0)\in \mathcal S$, $F:\mathcal S \times \mathcal W \to \mathcal S$ a Borel measurable function, and 
$(w_t)_{t\in \mathbb N}: \Omega \to (\mathcal{W})^{\mathbb{N}}$ a random variable sequence defined on the probability space $(\Omega, \mathcal F, \mathbb P)$. Assume that for any $\omega \in \Omega$, and indices $i\neq j$, the marginal variables $w_i$ and $w_j$ are independent.


We choose a Lyapunov function as a sequence of measurable functions $\{V_t: \mathcal S \to \mathbb R^+\}_{t\geq 0}$ such that $V_t(x) \geq 0$ for any $t\geq 0$, $x\in \mathcal S$. 
We say that the Lyapunov function $\{V_t\}_{t\geq 0}$ is radially unbounded if 
\[
\liminf_{x\in\mathcal S, \|x\|\to \infty} V_t(x) = \infty,
\]
for all $t$.
Define the distance of $x\in \mathbb R^n$ to a set $\mathcal M\subset \mathbb R^n$ by
\[
\dist(x,\mathcal M)=\inf_{p\in \mathcal M} \|x-p\|.
\]
Then we have the following invariance theorem showing that with certain conditions of $\{V_t\}$, every solution of the stochastic system  \eqref{eq:stoch_system} approaches an invariant set as $t\to \infty$ almost surely.
\begin{lemma}[Developed from Prop. 3.1 in \cite{zhang2016lasalle}]\label{lem:invariance_sto_system}
If there exists a radially unbounded Lyapunov function $\{V_t\}_{t\geq 0}: \mathcal S \to \mathbb R^+ $, and a real number $c\in (0,1)$ such that
\begin{equation}
\mathbb E[V_{t+1}(F(x, w_t))] \leq c V_t(x), \quad \forall x\in \mathcal S, t\in \mathbb N,
\end{equation}
then for any initial condition $x(0)\in \mathcal S$, the solution $\{x(t)\}_{t\geq 0}$ of  dynamics \eqref{eq:stoch_system}  fulfills 
\[\lim_{t\to \infty} \dist(x(t), \mathcal M) =0, \quad a.s., \]
where $\mathcal M$ satisfies 
\[
\mathcal M=\bigcup_{i=0}^\infty \bigcap_{j\geq i}^\infty \big\{x\in \mathcal S: V_j(x)=0 \big\}.
\]
\end{lemma}
\begin{proof}
See Appendix~\ref{sec:appendix_proof}.
\end{proof}

\subsection{Secret sharing schemes}
The \textit{secret sharing} schemes are 
encryption methods for sharing a confidential message with multiple parties, such that even with the collusion of a certain number of parties, the message should still not be  disclosed. 

In general, an $(n,p)$ secret sharing scheme often consists of two algorithms called by $\big($\textit{Share}, \textit{Reconstruct}$\big)$ with the forms that
\begin{itemize}
\item \textit{Share} takes as input a confidential message $S$ and outputs $n$ secret shares $\{S_1,\dots, S_n\}$;
\item \textit{Reconstruct} takes as input $p$ different secret shares $\{S_i\}_{i\in \mathcal I}$ for any index set $\mathcal I \subset [n]$ with $|\mathcal I|=p$, and outputs $S$.
\end{itemize}
Then each generated share $S_i$  is   distributed to one party $i$ for $i\in[n]$. Moreover, the security of secret sharing requires that any collusion of less than $p$ parties cannot reveal the message $S$.  
Secret sharing schemes have been used in many applications, e.g., encryption keys, distributed storage, missile launch codes, and numbered bank accounts (see \cite{al2011study,beimel2011secret, iftene2006secret, schoenmakers1999simple}, and the references therein). In these applications, each of the generated pieces of information must keep the original message confidential, as their exposure is undesirable, however, it is also critical that the message should be reconstructable.

One celebrated secret sharing scheme is the Shamir's scheme proposed by Adi Shamir \cite{shamir1979share}, in which the secret shares are generated by evaluation of a $(p-1)$-order polynomial at $n$ different points.
In particular, such a scheme satisfies information-theoretic security, i.e., even any $p-1$ secret shares reveal absolutely no information about the secret $S$.  
Inspired by the Shamir's scheme, in this paper, we will propose a privacy-preserving algorithm that employs polynomials to mask the true state of each agent in network consensus.


\section{Problem Formulation}\label{sec:prob_formulation}
We consider the problem of synchronizing  agents' states in a network, while avoid the risk of disclosing the true state of each agent to their neighbors.
Based on \textit{secret sharing} schemes, two algorithms are proposed to address such a problem: 
One is for the case where all the agents are required  to have a same security level, i.e., have a common maximum number of neighbors allowed to  collude to  pry into the agent's privacy without a security breach.
With the proposed algorithm in this scenario, average consensus  is achieved  eventually for all the agents.
Another algorithm is applicable to the case where the agents  can have different security levels. The price to pay is that the achieved consensus is not necessary being the average of the initial state any more.
In addition, it shows that both proposed algorithms  are also able to guarantee  communication safety, and be robust to a node failure in network consensus.

%

\subsection{Privacy degree}
This subsection is devoted to provide a detailed definition for the privacy security  employed in this paper.

We recall that in  network systems a distributed algorithm in general has the form
\begin{equation}\label{eq:general_algo}
  x_i(t+1) =  \mathcal F_i \Big(\,x_i(t), \, \big\{\mathcal R_t^{ij} \big(x_j(t)\big) \big\}_{j\in \mathcal N_i}, t \Big),
\end{equation}
for $i\in \mathcal V$, where $x_i(t) \in \mathbb R^d$ is the agent $i$'s state,
$\mathcal F_i(\cdot,\cdot, \cdot)$ is an iterative update law,  $\mathcal G=(\mathcal V, \mathcal E)$ is the inter-agent topology, 
 $\mathcal R_t^{ij}(x_j(t))$ is the information that agent $j$ sends to agent $i$ through the communication link $(i,j)\in \mathcal E$,
and
mapping $\mathcal R_t^{ij}: \mathbb R^d \to \mathbb R^m$ is agent $j$'s encoding function on edge $(i,j)$. 
Due to the fact that the update law $\mathcal F_i$ depends on each agent's own state and the information received from its neighbors, the algorithm can run in a distributed manner.

Next, for a  distributed algorithm given in \eqref{eq:general_algo}, we define its \textit{privacy degree}  to indicate each agent's capacity of ensuring privacy in the algorithm.

\begin{definition}[\textbf{Privacy degree}]\label{def:privacy_degree}
Given integer vector $p=(p_1, p_2, \dots, p_N) \in \mathbb N^N$.
We say that a distributed algorithm \eqref{eq:general_algo} has a privacy degree $p$, if for each $i\in \mathcal V$ and $t\in \mathbb N$, it holds that  the messages $\{\mathcal R_t^{\ell i}(x_i(t))\}_{\ell \in \mathcal I}$ do not disclose the privacy $x_i(t)$, for any index set $\mathcal I \subset \mathcal N_i$ with $| \mathcal I| \leq p_i$.
\end{definition}


The above definition entitles a distributed algorithm  with the ability that even a certain number of an agent $i$'s neighbors collude or the communication messages sent by it are disclosed, the agent $i$'s privacy $x_i(t)$ can still stay confidential. The following remark provides more insights about this definition.

\begin{remark}
For an algorithm of privacy degree $p$,
(i) if $p_i>0$, the state $x_i(t)$ stays confidential to each neighbor $j\in \mathcal{N}_i$. 
(ii) If $p_i=|\mathcal N_i|-1$, then $x_i(t)$ is disclosed only when all neighbors of agent $i$ collude. (iii) If $p_i \geq |\mathcal N_i|$, the state $x_i(t)$ is totally confidential in the network. 
\end{remark}
\noindent Furthermore, we note that in an algorithm with positive private degree, the communication links are also protected from eavesdropping, i.e., for any agent $i$ with a privacy degree $p_i >0$, an adversary wiretapping on any $p_i$ or less communication links in edges $\{(i,j): j\in \mathcal N_i\}$ learns nothing about $i$'s privacy $x_i(t)$.


In the existing methods, where the private state is masked by adding  deterministic offsets that vanish in average, e.g., \cite{gupta2017privacy}, or the accumulation of the added offsets goes to $0$ as $t\to \infty$, e.g., \cite{manitara2013privacy,rezazadeh2018privacy}, the privacy degree $p_i=|\mathcal{N}_i|-1$, $\forall i\in\mathcal V$, i.e., the privacy is preserved only if at least one neighbor is not malicious.

It is not hard to see that Definition~\ref{def:privacy_degree} only indicates an upper bound for the  number of neighbors allowed to be attacked or collude without a privacy breach. In order to make this bound tight,  we introduce the following definition of \textit{exact privacy degree} for a distributed algorithm.

\begin{definition}[\textbf{Exact privacy degree}]\label{def:strict_privacy_degree}
A distributed algorithm \eqref{eq:general_algo} has an exact privacy degree $p\in \mathbb N^N$, if 
\begin{enumerate}
\item[(i)] it  has a privacy degree $p$, and
\item[(ii)] for each $i\in\mathcal V$, $t\in \mathbb N$, and any index set $\mathcal I \subset \mathcal N_i$, if  $| \mathcal I| \geq p_i+1$,
the state $x_i(t)$ can be reconstructed from the  messages $\{\mathcal R_t^{\ell i}(x_i(t))\}_{\ell \in \mathcal I}$.  
\end{enumerate}
\end{definition}

In Definition~\ref{def:strict_privacy_degree}, the condition (ii)
not only provides the tightness of the degree of privacy, but also renders the network resilient to a node failure,
for example, the state of node $i$ is lost due to a blackout, or the node $i$'s connections are permanently broken.
More specific, at any moment when node $i$ fails, the reconstructability  requirement (ii) preserves the possibility of rebuilding the lost node via the information kept in enough number of its neighbors, even though none of the neighbors knows the exact state of node $i$. As a result, we typically want to set the exact privacy degree $p_i < |\mathcal{N}_i|$ to maintain such resilience to a node failure.


\subsection{Problem of privacy preserving consensus}

Now we are ready to raise the problem of privacy-preserving consensus that will be solved in this paper.

\begin{problem}\label{prob:1} 
In a network consisting of $N$ agents, the \textit{privacy-preserving consensus} with an exact privacy degree $p \in \mathbb N^N$ is to achieve that 
\begin{enumerate}
\item[(i)] the consensus algorithm has an exact privacy degree $p$, and
\item[(ii)]
asymptotic consensus is reached, i.e., for any initial condition $\{x_i(0)\}_{i\in \mathcal V}$, it holds that
\begin{equation}\label{eq:consensus_in_problem}
\lim_{t\to \infty} x_i(t)=x_\infty,\qquad \forall i\in[N],
\end{equation}
where $x_\infty$ is given by some non-trivial function of the initial states 
\footnote{Here we insist that $x_\infty$ can be expressed as a non-trivial function of the initial states, i.e., $x_\infty = G(x_1(0), \dots, x_N(0))$, for some given function $G: \mathbb{R}^{Nd} \to \mathbb{R}^d$. 
Because otherwise, there are trivial solutions requiring no information exchange at all. For example, when function $G(\cdot) = c $ with some pre-given constant $c$, consensus can be solved by directly setting each node's state $x_i(t)=c$.}.
\end{enumerate}  
\end{problem}

We note that in Problem~\ref{prob:1} the privacy-preserving requirement (i)  only matters before the consensus is achieved. 
After that, although the information transmitted amongst agents still keeps confidential, 
the agent state $x_i(t)$ is nevertheless already known to all the agents due to state consensus. 

Moreover,  we say that an algorithm reaches \textit{average consensus} with an exact privacy degree $p$, 
if in Problem~\ref{prob:1}  the 
function $x_\infty = \frac{1}{N} \sum_{i=1}^N x_i(0)$. 
In our previous work \cite{zhang2020consensus}, we address average consensus problem with a privacy degree instead of an exact privacy degree, i.e., there is no guarantee that the privacy $x_i(t)$ can be reconstructed from many enough messages $\{\mathcal R_t^{\ell i}(x_i(t))\}_{\ell \in \mathcal I}$.

\section{Algorithms and main results}\label{sec:alg_n_results}
The idea for solving Problem~\ref{prob:1} is to employ an \textit{encoding function} that splits the agent's local state $x_i(t)$ into $M$ distinct \textit{secret shares}, and sending each of these distinct shares to one of its neighbors $j \in \mathcal{N}_i$. This ensures that no single neighbor has complete knowledge of the privacy $x_i(t)$.

As the algorithm is expected to have an \textit{exact privacy degree} of $p$, 
we must further require that any $p_i$  or less of these shares leak no information about the secret $x_i(t)$, but $x_i(t)$ becomes totally reconstructable given any $(p_i+1)$ number of these shares.
Inspired by Shamir's secret sharing scheme, 
we use polynomials as the encoding functions, which are both simple and computationally inexpensive. 

In the rest of paper, for succinct expression, we will without loss of generality assume that the state dimension $d=1$. We note that in the paper  all the algorithms and results obtained can  be easily extended to higher-dimensional cases, (see Remark~\ref{rem:high_dim}).

\subsection{Encoding functions}

Specifically, the \textit{encoding function} used by agent $i$ is given by a polynomial of degree $p_i$,
\begin{equation}\label{eq:encoding_polynomial}
f_i(\theta,t)=\sum_{\ell=1}^{p_i} a_{i\ell}(t) \theta^\ell + x_i(t),
\end{equation}
for $\theta\in \mathbb R$ and $t\in \mathbb N$, where $a_{i\ell}(t) \in \mathbb R$ is some coefficient held by agent $i$, for $\ell \in[p_i]$, and $x_i(t)$ is agent $i$'s state. 
Notice that the state  $x_i(t)$ of agent $i$ is given by $f_i(0, t)$.
Since the coefficients $\{a_{i\ell}(t)\}_{\ell\in[p_i]}$ and state $x_i(t)$ are only known by agent $i$, the encoding function $f_i(\cdot, \cdot)$ is also kept secret from the other agents.

Given a set of communication keys $\{s_1, \dots, s_M\} \subset \mathbb Z^+$, each agent $i\in\mathcal V$ can use its encoding polynomial \eqref{eq:encoding_polynomial} to generate $M$ encoded states  $\{r_i^k(t)\}_{k\in [M]}$ by setting
\begin{equation}\label{eq:rik_encoding_function}
r_i^k(t) = f_i(s_k, t), \qquad \forall k\in[M],
\end{equation}
where $k$ is called a channel index, and $r_i^k(t)$ the encoded state on channel $k$.
Since the degree of the polynomial $f_i(\cdot,t)$ is $p_i$, it is possible to reconstruct the local state $x_i(t)$ from any combination of $p_i+1$ distinct encoded states $r_i^k(t)$, provided that $p_i < M$.


From now on, let  $r_i(t) = (r_i^{1}(t), \dots, r_i^{M}(t)) \in \mathbb R^M$, $i\in \mathcal V$, and $r^{k}(t) = (r_1^{k}(t), \dots, r_N^{k}(t))\in \mathbb R^{N}$, for any $k\in[M]$. Furthermore, let
\begin{equation}\label{eq:dmax}
d_{max} = \max_{i\in\mathcal V} |\mathcal N_i|,
\end{equation} 
denote the maximal degree in the graph $\mathcal G$.


\subsection{Channel selection}\label{sec:channel_selection}
At each time $t\in \mathbb N$, all the agents will first launch a handshake procedure, through which each edge $(i,j)\in \mathcal E$ is assigned with a channel $c_{ij}(t)\in [M]$.
Here $c_{ij}(t)$ is a channel indicator. This handshake procedure is presented in Algorithm~\ref{alg:handshake1}, and can be performed in an asynchronous, decentralized manner. 

The next lemma shows that the handshake Algorithm~\ref{alg:handshake1} fulfills that all adjacent edges can be assigned with mutually different channels. 

\begin{lemma}\label{lem:property_Handshake}
If $M \geq 2 d_{max} -1 $, at any time $t\in \mathbb N$,  Algorithm~\ref{alg:handshake1} achieves that
\begin{enumerate}
\item[(i)] assign every edge with a channel, i.e.,  $c_{ij}(t)\in [M]$, $\forall (i,j) \in \mathcal E$, 
\item[(ii)] adjacent edges occupy different channels, i.e.,  the assigned channel $c_{ij}(t)$ satisfies
\begin{equation}\label{eq:cij}
 c_{ij}(t) \notin \bigg(\bigcup_{k\in \mathcal N_i \setminus \{j\}} \{c_{ik}(t)\} \bigg) \bigcup
 \bigg(\bigcup_{k\in \mathcal N_j \setminus \{i\}} \{c_{jk}(t)\}
 \bigg),
\end{equation}
for any $(i,j) \in \mathcal E$.
\end{enumerate}
\end{lemma}
\begin{proof}
To prove (i), we only need show that in Algorithm~\ref{alg:handshake1} the set $K=[M] \setminus (R_i \cup R_j)$ is nonempty at any time, where $R_i$ is all the channels that have been assigned to the edges incident  to  $i$. 
When we assign a channel to an unassigned edge $(i,j) \in \mathcal E$, there are at most $(d_{max}-1)$ edges that have already been assigned in the edges incident to $i$,  and at most $(d_{max}-1)$ assigned edges in edges incident to $j$.
 Because $M > 2 d_{max} -2 $, set $K$ is nonempty. The argument (ii) 
 follows the fact that each channel $c_{ij}(t)$ is chosen from the set $\Omega=[M]\setminus (R_i \cup R_j)$. 
 \end{proof}
 
\renewcommand{\algorithmicrequire}{\textbf{Input:}}
\renewcommand{\algorithmicensure}{\textbf{Output:}}
\algsetup{indent=12pt}

\begin{algorithm}[t]
\caption{Handshake: channel selection at time $t$.}
\label{alg:handshake1}
\begin{algorithmic}[1]
\REQUIRE Graph $\mathcal G=(\mathcal E, \mathcal V)$; Number of channels $M$ satisfying $M \geq 2 d_{max}-1$.

\STATE \textbf{Initialization:} For each edge $(i,j)\in \mathcal E$, initialize  channel indicator $c_{ij}(t)=0$.\footnotemark
\,
Set $\mathcal A_i =  \mathcal N_i$ be the set of $i$'s neighbors connected with an unassigned edge.

\STATE \textbf{Each agent $i$ preforms the following process in parallel:} 
\begin{enumerate}[nolistsep]
   \item[\textbf{Step 1}] (Edge selection):\\  Find the neighbor $j =  \min_{\ell} \{\ell\in \mathcal{A}_i\}$.\footnotemark \\
Ping node $j$, and wait until $j$ pings back.\footnotemark
   
   \item[\textbf{Step 2}] (Connection establishment):\\ 
   Once $j$ pings back, agents $i,j$ together choose an available channel $k$ uniformly at random  from the set $\big( [M] \setminus (R_i \cup R_j)\big)$ for the edge $(i,j)$, where
    \begin{align*}
	    R_i &= \big(\bigcup_{k\in \mathcal N_i} \{c_{ik}(t)\}\big) \setminus \{0\}, \\
	    R_j &= \big(\bigcup_{k\in \mathcal N_j} \{c_{jk}(t)\}\big) \setminus \{0\}.
	    \end{align*}
	 Set $c_{ij}(t)=c_{ji}(t)=k$.
	 Remove $j$ from $\mathcal A_i$, and remove $i$ from $\mathcal A_j$.
	 
	 \item[\textbf{Step 3}] (Repeat):\\ If $\mathcal A_i$ is non-empty, go to \textbf{Step 1}. If it is empty, the  agent $i$'s handshake is done.
\end{enumerate}
\STATE \textbf{END}: Handshake is done when all agents finish their own handshake processes.

\end{algorithmic}
\end{algorithm}
\addtocounter{footnote}{-2}
\footnotetext{The initialization $c_{ij}(t)=0$ indicates that an edge $(i,j)$ has not been assigned with a channel yet.}
\addtocounter{footnote}{+1}
\footnotetext{Note that set $\mathcal A_i$ is not empty whenever there exist edges in $\mathcal{N}_i$ not being assigned.}
\addtocounter{footnote}{+1}
\footnotetext{Ping is an operation that a node sends a specific message  to its neighbor to notify that it is ready to establish a connection.}
 
\begin{remark}
In Step 1 of Algorithm~\ref{alg:handshake1}, we assume all the agents are ordered according to their index $i$. Alternatively, an order of agents can be specified by any other nodes' identification, e.g., the MAC address on the internet.
\end{remark}

As to the complexity of Algorithm~\ref{alg:handshake1}, suppose we consider Step 2  as the elementary operation, then the time complexity of the whole algorithm is $\mathcal{O}(N^2)$. This is because the number of edges $| \mathcal E| = \frac{1}{2}\sum_{i = 1}^N | \mathcal N_i| \leq N^2$. Therefore, it is guaranteed that the handshake procedure of all the agents can be finished in a bounded time.

\subsection{Synchronization update law}

We are now ready to give the update law for privacy-preserving consensus. The basic idea is to, for each separate channel $k\in[M]$, first perform an average consensus step for the encoded states belonging to that channel, through the edges  with channels assigned in Section~\ref{sec:channel_selection}, and then conduct a projection operator to update local states $x_i$ and coefficients $a_{i}$.

At each time $t$, once the handshake process in Algorithm~\ref{alg:handshake1} for an agent $i\in \mathcal V$ is finished, i.e., $c_{ij}(t) \neq 0$ for all $j\in \mathcal N_i$, the state of agent $i$ is then updated according to the update law
\begin{subequations} \label{eq:update_law_total}
\begin{align}
\widetilde{r}^{k}_i (t+1) &= r^{k}_i(t) \!+\! \gamma\!\! \sum_{j\in \mathcal N_i}  l_{ij}^{k}(t)\Big(r_j^{k}(t) - r_i^{k}(t)\Big), \; \forall k\in [M], \label{eq:update_law_a}\\
\begin{bmatrix}
x_i(t+1)\\
a_i(t+1)
\end{bmatrix}
&= T_i \;\widetilde{r}_i (t+1),\label{eq:update_law_b}\\
r_i^k(t+1) &= f_i(s_k, t+1), \; \forall k\in [M], \label{eq:update_law_c}
\end{align}
\end{subequations}
%
with the initial condition $r_i^k(0)= f_i(s_k, 0)$ for every   $k\in [M]$, and the step size $\gamma>0$. Here, the communication weight $l_{ij}^k(t)= \mathbbm{1}(c_{ij}(t)=k)$, 
vector $\widetilde{r}_i(t) = (\widetilde{r}_i^{1}(t), \dots, \widetilde{r}_i^{M}(t)) \in \mathbb R^M$,
$a_i(t) = (a_{i1}(t), \dots, a_{ip_i}(t)) \in \mathbb R^{p_i}$
and $f_i(\cdot, \cdot)$ is the encoding polynomial  defined in \eqref{eq:encoding_polynomial}. 
Furthermore, the projection matrix $T_i=(\Phi_i^T \Phi_i)^{-1}\Phi_i^T \in \mathbb R^{(p_i+1) \times M}$ with $\Phi_i$ being the Vandermonde matrix related to the key sequence $S = (s_1, \dots, s_M) \in \mathbb R^M$, i.e.,
\begin{equation}\label{eq:Vandermonde}
\Phi_i=
\begin{bmatrix}
1 & s_1 & (s_1)^2 & \dots & (s_1)^{p_i}\\
1 & s_2 & (s_2)^2 & \dots & (s_2)^{p_i}\\
\vdots & \vdots &\vdots & \ddots & \vdots\\
1 & s_M & (s_M)^2 & \dots & (s_M)^{p_i}
\end{bmatrix} \in \mathbb R^{M \times (p_i+1)}.
\end{equation}
Note that if  $s_i \neq s_j, \forall i\neq j$ and $M \geq p_i+1$, then $\rank(\Phi_i) = p_i+1$.

The following remark indicates that  the individual privacy is not disclosed  in the proposed algorithm \eqref{eq:update_law_total}.
\begin{remark}
Due to the definition of $l_{ij}^k(t)$, in the update law \eqref{eq:update_law_total}, each agent $i$ only accesses to  one of neighbor $j$'s encoded states, i.e., $r_j^{c_{ij}(t)}(t)$.
\end{remark}

More specifically, the update  \eqref{eq:update_law_a} can be seen as a  channel-wise consensus  process, by which  encoded state $r_i^k(t)$ is averaged among the $i$'s neighbors being assigned with 
channel $k$ at each time $t$.  
During this consensus procedure, 
in order to maintain the updated encoded state $r_i(t+1)$ being able to generate from a polynomial of order $p_i$, 
\eqref{eq:update_law_b} is used to project the intermediate consensus variable $\widetilde{r}_i (t+1)$ back to the subspace $\Span (\Phi_i)$. 
Note that $\Phi_i$ is full column rank, and
the update \eqref{eq:update_law_b} indeed solves the optimization problem given by
\[
\begin{bmatrix}
x_i(t+1)\\
a_i(t+1)
\end{bmatrix}
= \argmin_{v\in \mathbb R^{p_i+1}} \big\|\widetilde{r}_i (t+1) - \Phi_i v \big\|_2,
\]
where $\|\cdot\|_2$ is the Euclidean norm.

Next, we have the following lemma to show that 
Algorithm~\eqref{eq:update_law_total} has an exact privacy degree $p$, and more essentially provides  a method to reconstructing the agent $i$'s privacy from a certain number of its encoded states.

\begin{lemma}\label{lem:recon_unrecon_poly}
The distributed update law \eqref{eq:update_law_total} has an exact privacy degree $p\in \mathbb N^N$. Moreover, for any $i\in \mathcal V$ and an arbitrary  index set $\mathcal I \subset [M]$ satisfying $|\mathcal I|\geq p_i+1 $, the privacy $x_i(t)$ can be reconstructed by
\begin{equation}\label{eq:recon_poly_larger_than_pi}
x_i(t)=\sum_{k\in \mathcal I} r_i^k (t) \frac{\prod_{\ell\in \mathcal I \setminus \{k\}}  (-s_\ell)}{ \prod_{\ell\in \mathcal I \setminus \{k\}}  (s_k-s_\ell)  }, \quad \forall  t\in \mathbb N.
\end{equation}
\end{lemma}

\begin{proof}
First we need show that for any $i \in \mathcal V$, the messages  $\{r_i^k(t)\}_{k \in \mathcal I}$ do not
disclose $x_i(t)$, if the index set $\mathcal I \subset \mathcal N_i$ satisfies $| \mathcal I| \leq p_i$.
This is because even an adversary know $p_i$ encoded state $\{r_i^k(t)\}_{i\in \mathcal I'}$, for each state candidate  $x_i'(t)$ the adversary can construct one and only one polynomial $f'(\theta,t)$ of degree $p_i$ such that $f'(0,t)=x'_i(t)$ and $f'(s_k,t)=r^k_i(t)$, $k\in \mathcal I'$. Moreover, each of these polynomials are equally likely, therefore the adversary cannot learn 
the real privacy $x_i(t)$.

The reconstruction formula \eqref{eq:recon_poly_larger_than_pi} follows the interpolation polynomial in the Lagrange form.
\end{proof}

Now it is ready to present the main theorems related to  convergence of the proposed privacy-preserving consensus scheme, including Algorithm~\ref{alg:handshake1} and update law \eqref{eq:update_law_total}.
We consider two cases separately: one is when all the agents have a same privacy degree, and another is for the case that the privacy degree is individually set.

When all the agents share an identical privacy degree, the next theorem shows that the agents can reach average consensus by the proposed method.

\begin{thm}[Common privacy degree]\label{thm:common_degree}
Let  $p= p_0 \mathbf{1}_N$ for a given $p_0 \in \mathbb N$. If the step size $\gamma\in(0,1)$, and $M \geq 2 d_{max}-1$,
then update law \eqref{eq:update_law_total} with the handshake procedure in Algorithm~1 almost surely solves Problem~\ref{prob:1} with average consensus reached, i.e.,
the algorithm has an exact privacy degree $p$, and for any initial condition $\{x_i(0)\}_{i\in \mathcal V}$,   
\begin{equation}\label{eq:consensus_in_thm_common_degree}
\lim_{t\to \infty} x_i(t)=\frac{1}{N}\sum_{j=1}^N x_j(0),\qquad \forall i\in[N],
\end{equation}
almost surely.
\end{thm}

\begin{remark}
In Theorem~\ref{thm:common_degree}, the almost sure convergence is with respect to 
 probability stemming from random channel selections in Algorithm~\ref{alg:handshake1}. 
\end{remark}

Next, we continue with a more general case where agents are allowed to have distinct privacy degrees. 
In this case, the algorithm proposed still guarantees that  agents reach consensus almost surely with the desired privacy requirement, but the value
agreed upon may not necessarily be the average of the initial condition $\{x_i(0)\}_i$.

\begin{thm}[Different privacy degree]\label{thm:dif_degree}
Let vector $p\in \mathbb N^N$. If  $\gamma\in(0,1)$ and $M \geq 2 d_{max}-1$,  update law \eqref{eq:update_law_total} with handshake Algorithm~1 almost surely solves Problem~\ref{prob:1},  i.e.,
the algorithm has an exact privacy degree $p$, and for any initial condition $\{x_i(0)\}_{i\in \mathcal V}$,   
\begin{equation}\label{eq:consensus_in_thm_common_degree}
\lim_{t\to \infty} x_i(t)=x_\infty,\qquad \forall i\in[N],
\end{equation}
almost surely, for some real $x_\infty$.
\end{thm}

\begin{remark}
In Theorem~\ref{thm:dif_degree}, the steady state $x_\infty$ can be computed by equation \eqref{eq:x_inf_different_degree}, (see also Remark~\ref{rmk:r_inf_different_degree}).
\end{remark}

\section{Convergence analysis}\label{sec:converge}
\subsection{Case of common privacy degree}\label{sec:converge_common_degree}
First, according to update law \eqref{eq:update_law_total}, we can derive the stacked encoded state is governed by following dynamics
\begin{subequations} \label{eq:stack_update_ki_total}
\begin{align}
\widetilde{r}^k(t+1) &= r^k(t) - \gamma L^k(t) r^k(t),\quad \forall k\in[M]\\
r_i(t+1) &= \overline{T}_i\widetilde{r}_i(t+1), \qquad  \forall k\in \mathcal V  
\end{align}
\end{subequations}
where the projection matrix $\overline{T}_i=\Phi_i(\Phi_i^T \Phi_i)^{-1}\Phi_i^T=\Phi_i T_i \in \mathbb R^{M\times M}$, vector $\widetilde{r}^{k}(t) = (\widetilde{r}_1^{k}(t), \dots, \widetilde{r}_N^{k}(t))\in \mathbb R^{N}$, $k\in[M]$, and the entry of matrix $L^k(t) \in \mathbb R^{N\times N}$ is 
\begin{equation}\label{eq:def_Lk(t)}
(L^k(t))_{ij}=\begin{cases}
-l_{ij}^k(t), \quad &\text{if } i\neq j,\\
\sum_{\ell\in \mathcal{N}_i} l_{i\ell}^k(t), \quad &\text{otherwise}. 
\end{cases}
\end{equation}
Due to the above definition, the matrix $L^k(t)$ is doubly stochastic for any $k\in[M], t\in \mathbb N$.

Moreover, as the weight $l_{ij}^k(t)=\mathbbm{1}(c_{ij}(t)=k)$ depending on the random channel selection in the handshake procedure given by Algorithm~\ref{alg:handshake1}, the  matrix $L^k(t)$ and encoded states $r_i^k(t)$ are therefore random variables for each $k\in[M], i\in \mathcal V, t\in \mathbb N$. Then the random matrix $L^k(t)$ satisfies the following lemma.

\begin{lemma}\label{lem:expactation_Lk(t)}
The random matrix $L^k(t)$ constructed by Algorithm~\ref{alg:handshake1} satisfies that
\begin{subequations}
\begin{align}
\mathbb{E}\left[L^k(t)\right] &= \frac{1}{M}\,L \, , \label{eq:expaction_Lk(t)}\\
\mathbb{E}\left[\big(L^k(t)\big)^2\right] &= \frac{2}{M}\,L \,, \label{eq:expaction_Lk(t)2}
\end{align}
\end{subequations}
for any $k\in [M]$, $t\in \mathbb N$, where $L$ is the Laplacian matrix of the connectivity graph $\mathcal G$, and $M$ is the number of channels.
\end{lemma}

\begin{proof}
First, for any channel selections $\{c_{ij}(t)\}_{(i,j)\in\mathcal E}$ given by Algorithm~\ref{alg:handshake1}, we have $c_{ij}(t) \sim U\{[M]\}$, where  $U(A)$ is  the uniform distribution supported on a finite set $A$. This is due to the channel selection is invariant under any permutation among channels. Equivalently, it holds that $\mathbb P\{l^k_{ij}(t)=1\}=\frac{1}{M}$, for any $k\in[M]$, $(i,j)\in \mathcal E$. Therefore, we have for each $k$
\[
\mathbb{E}[l^k_{ij}(t)]=\begin{cases}
\frac{1}{M}, \quad &\text{if } (i,j)\in \mathcal E,\\
0, \quad &\text{otherwise}. 
\end{cases}
\]
According to definition \eqref{eq:def_Lk(t)}, the assertion \eqref{eq:expaction_Lk(t)} follows.

To prove \eqref{eq:expaction_Lk(t)2}, we notice that according to Lemma~\ref{lem:property_Handshake}, Algorithm~\ref{alg:handshake1} assigns any adjacent edges with distinct channels. This further implies that  except the diagonal entries, there is at most one nonzero entry in each row of matrix $L^k(t)$. So is in each column of matrix $L^k(t)$, since $L^k(t)$ is symmetric.

In order to compute the entries of  $\big(L^k(t)\big)^2$, we consider the following  cases:
\begin{enumerate}
\item If $i\neq j$ and $c_{ij}(t)=k$, 
\begin{align*}
\left( \big(L^k(t)\big)^2 \right)_{ij} &= \sum_{\ell=1}^N \left(L^k(t)\right)_{i\ell} \left(L^k(t)\right)_{\ell j}\\
 &= \left(L^k(t)\right)_{ii}\! \left(L^k(t)\right)_{i j}\!+\!\left(L^k(t)\right)_{ij} \!\left(L^k(t)\right)_{j j}\\
 &=-2.
\end{align*}
\item If $i\neq j$ and $c_{ij}(t)\neq k$, 
\begin{align*}
\left( \big(L^k(t)\big)^2 \right)_{ij} &= \left(L^k(t)\right)_{ii}\! \left(L^k(t)\right)_{i j}\!+\!\left(L^k(t)\right)_{i\ell} \!\left(L^k(t)\right)_{\ell j}\\
&=0,
\end{align*}
where we use the fact that in $i$-th row only the diagonal entry and at most one else entry (w.l.g we assume it is the $\ell$-th entry) could be nonzero.
\item If $i=j$ and there exists $\ell$ such that $c_{i\ell}(t) =k$, 
\begin{align*}
\left( \big(L^k(t)\big)^2 \right)_{ij} 
&= \left(L^k(t)\right)_{ii}\! \left(L^k(t)\right)_{i j}\!+\!\left(L^k(t)\right)_{i\ell} \!\left(L^k(t)\right)_{\ell j}\\
&=2.
\end{align*}
\item If $i=j$ and $c_{i\ell}(t) =0$, $\forall \ell$, we can show  $\left( \big(L^k(t)\big)^2 \right)_{ij}=0$.
\end{enumerate}

Hence, taking expectation of each entry in matrix $\big(L^k(t)\big)^2$ gives 
\[
 \mathbb{E}\left[ \left( \big(L^k(t)\big)^2\right)_{ij} \right]
=\begin{cases}
2\frac{|\mathcal{N}_i|}{M}, \quad &\text{if } i= j,\\
-\frac{2}{M}, \quad &\text{if } i\neq j, (i,j)\in \mathcal E,\\
0, \quad &\text{otherwise},
\end{cases}
\]
which proves \eqref{eq:expaction_Lk(t)2}.
\end{proof}

Denote by $r(t)=\left(r^1(t)^T, r^2(t)^T, \dots, r^M(t)^T \right) \in \mathbb R^{MN}$ the stacked vector of all encoded states. Because all agents share a common privacy degree $p_0$,   in \eqref{eq:update_law_b} the project matrices of all the agents are also same, i.e., $T_i= T_0 = (\Phi_0^T \Phi_0)^{-1}\Phi_0^T$  for all $i\in \mathcal V$, where $\Phi_0$ is the  Vandermonde matrix of order $p_0$. 

Then, according to the channel-wise and agent-wise updates given in \eqref{eq:stack_update_ki_total}, we can derive the overall stacked state $r(t+1)$ satisfying
\begin{equation}\label{eq:total_dynamics_commonp}
r(t+1) = \big( \overline{T}_0 \otimes I_N  \big) \big( I_{MN} - \gamma \overline{L}(t)\big) r(t),
\end{equation}
where project matrix $\overline{T}_0 = \Phi_0 T_0$, and matrix $\overline L(t)$ has the form
\begin{equation}\label{eq:def_overlin_L}
\overline{L}(t)=\begin{bmatrix}
L^1(t) & & \\
& \ddots & \\
&  &    L^M(t) 
\end{bmatrix},
\end{equation}
with each diagonal block being $L^k(t)$ defined in \eqref{eq:def_Lk(t)}.

We note that in overall dynamics \eqref{eq:total_dynamics_commonp} the matrix $\overline{T}_0$ is symmetric and satisfies $\left(\overline{T}_0\right)^\ell=\overline{T}_0$, $\forall \ell\in \mathbb Z^+$, i.e., $\overline{T}_0$ 
is an idempotent matrix. 
Therefore, it holds that any eigenvalue of matrix $\overline{T}_0$ is either $0$ or $1$. Moreover, we can verify that  the matrix $\overline{L}(t)$  satisfies
\begin{subequations}
\begin{align}
\left(A\otimes \mathbf{1}_N^T \right)\, \overline{L}(t) &= 0_{M\times MN}, \quad \forall A\in \mathbb R^{M\times M}, \label{eq:prop_overline_L_1}\\
 \overline{L}(t)\, \left(v\otimes \mathbf{1}_N \right) &= 0_{MN\times 1}, \quad \forall v\in \mathbb R^M. \label{eq:prop_overline_L_2}
\end{align}
\end{subequations}
for any $t\in \mathbb N$.

Next, the following proposition shows that on each channel $k$  the encoded state of each agent converges almost surely  to the average of the encoded states at $t=0$.

\begin{prop}\label{prop:r_converge_to_average}
Under the same assumptions as in Theorem~\ref{thm:common_degree}, it holds that
\[
\lim_{t\to \infty} r^k_i(t) =  \frac{1}{N}\sum_{j=1}^N r^k_j(0), \quad a.s.,
\]
for each $i\in\mathcal V$, $k\in [M]$.

\end{prop}

\begin{proof}
We define $\overline{r}^k(t) = \frac{1}{N}\sum_{i=1}^N r^k_i(t)$ for $k\in [M], t\in \mathbb N$, and denote the stacked vector by $\overline{r}(t)=\left(\overline{r}^1(t), \dots, \overline{r}^M(t) \right) \in \mathbb R^{M}$. Then according to the dynamics \eqref{eq:total_dynamics_commonp}, we have
\begin{align*}
\overline{r}(t+1) &= \frac{1}{N}\left( I_M \otimes \mathbf{1}_N^T\right)r(t+1)\\
&= \frac{1}{N}\left( I_M \otimes \mathbf{1}_N^T\right) \big( \overline{T}_0 \otimes I_N  \big) \big( I_{MN} - \gamma \overline{L}(t)\big) r(t)\\
&= \frac{1}{N} \big( \overline{T}_0 \otimes \mathbf{1}_N^T  \big)  r(t)\\
&= \overline{T}_0\overline{r}(t),
\end{align*}
where we employ \eqref{eq:prop_overline_L_1} in the third equality. 
Moreover,  due to the fact  that $r_i(0) \in \Span\{\Phi_0\}$, we have $\overline{r}(0)=\frac{1}{N}\sum_{i=1}^N r_i(0) \in \Span\{\Phi_0\} $. Therefore, it holds that 
\[
\overline{r}(t)= \left(\overline{T}_0\right)^{t} \overline{r}(0)=\overline{r}(0),\]
for any $t\in \mathbb N$.

Next, define disagreement vector 
\begin{equation}\label{eq:def_delta_k}
\delta^k(t) = r^k(t) - \overline{r}^k(0) \mathbf{1}_N,
\end{equation}
for each $k\in[M]$, and we have $\mathbf{1}_N^T \delta^k(t)=0$.
Furthermore, the stack vector $\delta(t)=\left(\delta^1(t)^T,\dots, \delta^M(t)^T\right)\in \mathbb R^{MN}$ satisfies
\begin{align}\label{eq:dynamics_delta_common}
\delta(t+1) &= r(t+1) - (\overline{r}(0) \otimes \mathbf{1}_N) \nonumber\\
&= \big( \overline{T}_0 \otimes I_N  \big) \big( I_{MN} - \gamma \overline{L}(t)\big) r(t)  - (\overline{r}(0) \otimes \mathbf{1}_N)\nonumber\\
&= \big( \overline{T}_0 \otimes I_N  \big) \big( I_{MN} - \gamma \overline{L}(t)\big) \big[r(t)  - (\overline{r}(0) \otimes \mathbf{1}_N)\big]\nonumber\\
&=\big( \overline{T}_0 \otimes I_N  \big) \big( I_{MN} - \gamma \overline{L}(t)\big) \delta(t),
\end{align}
where the third equality is due to \eqref{eq:prop_overline_L_2} and the fact that $\overline{T}_0 \overline{r}(0)=\overline{r}(0)$.

For any  $z\in \mathbb R^{MN}$, denote vector $z^k \in \mathbb R^N$ such that there is a partition satisfying $z=\big[(z^{1})^T, (z^{2})^T, \dots, (z^{M})^T \big]^T$.
Then define  a subspace 
\begin{equation}\label{eq:subspace_S}
\mathcal S=\left\{z\in \mathbb{R}^{MN}: \mathbf{1}_N^T z^k(t)=0, k\in[M]\right\}.
\end{equation}
For the dynamics \eqref{eq:dynamics_delta_common}, taking the Lyapunov functions $V: \mathcal S \to \mathbb R^+$ defined by $V(z)=z^Tz$,  it holds that for any $z\in \mathcal S$
\begin{align}\label{eq:pf_E_V_1}
\,&\, \mathbb E\Bigg[  V\bigg( \big( \overline{T}_0 \otimes I_N  \big) \big( I_{MN} - \gamma \overline{L}(t)\big) z  \bigg)\Bigg]\nonumber\\
\leq &\,
\mathbb E\Bigg[  \bigg\| \big( I_{MN} - \gamma \overline{L}(t)\big) z  \bigg\|^2 \Bigg]\nonumber\\ 
=&
\sum_{k=1}^M (z^{k})^T	\mathbb E\Bigg[
\big( I_{N} - \gamma L^k(t)\big)^T\big( I_{N} - \gamma L^k(t)\big) \Bigg]  z^k \nonumber\\
=&
\sum_{k=1}^M (z^{k})^T	\bigg( I_N - \frac{2}{M}(\gamma-\gamma^2)L \bigg)  z^k,
\end{align}
where 
the first inequality follows the fact that $\sigma(\overline{T}_0) \subset \{0,1\}$, the last equality is due to Lemma~\ref{lem:expactation_Lk(t)}.

As the connectivity graph is connected, its Laplacian matrix $L$ has a simple
smallest eigenvalue $0$ with the eigenvector $\mathbf{1}_N$. Moreover, 
by the Courant-Fisher theorem, it holds that
\begin{align*}
\max_{\substack{z \in \mathbb R^N \setminus 0\\z^T\mathbf{1}_N=0 }}
\frac{z^T \bigg( I_N - \frac{2}{M}(\gamma-\gamma^2)L \bigg) z}{z^Tz}
= 1- \frac{2}{M} (\gamma-\gamma^2)\lambda_2,
\end{align*}
where $\lambda_2>0$ is the second smallest eigenvalue of matrix $L$. Because $\mathbf{1}^T_N z^k=0$, $\forall k\in[M]$ in \eqref{eq:pf_E_V_1}, we have
\begin{equation}\label{eq:pf_E_V_2}
\mathbb E\Bigg[  V\bigg( \big( \overline{T}_0 \otimes I_N  \big) \big( I_{MN} - \gamma \overline{L}(t)\big) z  \bigg)\Bigg] 
\leq
\overline{c}\, V(z), 
\end{equation}
for any $z\in\mathcal S$, where $\overline{c}=1- \frac{2}{M} (\gamma-\gamma^2)\lambda_2 $. Due to the ansatz that $\gamma\in (0,1)$, we can show that
\[
\overline{c} \in \big[1- \frac{\lambda_2}{2M} \, ,\, 1\big).
\]
In addition, by the Gershgorin theorem, the eigenvalue $\lambda_2$ is located in the interval $(0, 2 d_{max}]$. Combining this with  $M\geq 2d_{max}-1$, we have $\overline{c} \in [\frac{d_{max}-1}{2d_{max}-1},1)$.

In the end, by using Lemma~\ref{lem:invariance_sto_system}, we have in the disagreement dynamics \eqref{eq:dynamics_delta_common},
\[\lim_{t\to \infty} \dist(\delta(t), \mathcal M) =0, \quad a.s., \]
where $\mathcal M= \{z\in\mathbb R^{MN}:z=0\}$. According to \eqref{eq:def_delta_k}, this suggests that
\[
\lim_{t\to \infty} r^k_i(t) =  \frac{1}{N}\sum_{j=1}^N r^k_j(0), \quad a.s.,
\]
for each $i\in\mathcal V$, $k\in [M]$.

\end{proof}

Now we are ready to complete the proof of Theorem~\ref{thm:common_degree}.

\begin{proof}[\textbf{Proof of Theorem~\ref{thm:common_degree}}]

By Proposition~\ref{prop:r_converge_to_average}, we have for each $k\in [M]$ and $i\in\mathcal V$, 
\begin{align*}
\lim_{t\to \infty} r^k_i(t) 
&=  \frac{1}{N}\sum_{j=1}^N f_j(s_k,0), \quad a.s.\\
&=  \frac{1}{N} \sum_{j=1}^N \left( \sum_{\ell=1}^{p_0} a_{j\ell}(0) (s_k)^\ell + x_j(0) \right)\\
&= \sum_{\ell=1}^{p_0} \left(\frac{1}{N} \sum_{j=1}^N a_{j\ell}(0)\right) (s_k)^\ell + \left(\frac{1}{N} \sum_{j=1}^N  x_j(0)\right)\\
&=:f_\infty(s_k),
\end{align*}
where the polynomial $f_\infty(\theta)=\sum_{\ell=1}^{p_0} b_{\ell} \theta^\ell + b_0$ with $b_0=\frac{1}{N}\sum_j  x_j(0)$, and $b_\ell=\frac{1}{N} \sum_j a_{j\ell}(0)$.

Next, by using Lemma~\ref{lem:recon_unrecon_poly} with $\mathcal I=[M]$, we have for each $i\in \mathcal V$,
\begin{align*}
\lim_{t\to \infty} x_i(t)
&= \sum_{k\in \mathcal I} \lim_{t\to \infty} r_i^k (t) \frac{\prod_{\ell\in \mathcal I \setminus \{k\}}  (-s_\ell)}{ \prod_{\ell\in \mathcal I \setminus \{k\}}  (s_k-s_\ell)  }\\
&=\sum_{k\in \mathcal I} f_\infty(s_k) \frac{\prod_{\ell\in \mathcal I \setminus \{k\}}  (-s_\ell)}{ \prod_{\ell\in \mathcal I \setminus \{k\}}  (s_k-s_\ell)  } \quad a.s.\\
&= b_0,
\end{align*}
in which we use the Lagrange interpolation in the last equality.
\end{proof}

\subsection{Case of different privacy degree}
Compared to the case in  section~\ref{sec:converge_common_degree}, the various privacy degree $p\in \mathbb{N}^N$ makes the project matrices $T_i$ in \eqref{eq:update_law_b} different from the agents. 
Then in order to find the dynamics for the overall state $r(t)$ as in \eqref{eq:total_dynamics_commonp}, we introduce a  permutation matrix $G \in \mathbb{R}^{MN \times MN}$ defined by
\[
G=
\begin{bmatrix}
e_{\pi(1)}\\
e_{\pi(2)}\\
\vdots\\
e_{\pi(MN)}
\end{bmatrix},
\]
where $e_i$ is the $i$-th row of the identity matrix $I_{MN}$, and  mapping $\pi: [MN] \to [MN]$ is defined by  
\begin{equation}\label{eq:permutation_map}
\pi(\ell) = \left(\ell- M \floor*{\frac{\ell-1}{M}}-1 \right)N + \ceil*{\frac{\ell}{M}}.
\end{equation}
We note that $G$ is an orthogonal matrix.
Denote by $r_a(t) = \left(r_1(t)^T, r_2(t)^T, \dots, r_N(t)^T \right) \in \mathbb R^{MN}$ the stacked encoded state in the agent order. It can be verified that $r_a(t) =  G r(t)$.
\begin{example}
Take an example with $N=2$, $M=3$, then under the permutation mapping defined in \eqref{eq:permutation_map}, we have
\[
r_a(t)= G \, r(t) = 
\begin{bmatrix}
e_{\pi(1)}\\
e_{\pi(2)}\\
e_{\pi(3)}\\
e_{\pi(4)}\\
e_{\pi(5)}\\
e_{\pi(6)}
\end{bmatrix}
r(t)
=\begin{bmatrix}
e_{1}\\
e_{3}\\
e_{5}\\
e_{2}\\
e_{4}\\
e_{6}
\end{bmatrix}
 \begin{bmatrix}
r^1_1\\
r^1_2\\
r^2_1\\
r^2_2\\
r^3_1\\
r^3_2\\
\end{bmatrix}
=
\begin{bmatrix}
r^1_1\\
r^2_1\\
r^3_1\\
r^1_2\\
r^2_2\\
r^3_2\\
\end{bmatrix}.
\]
\end{example}

Then it can be derived that the dynamics for overall encoded state $r(t)$ with $p\in \mathbb{N}^N$ is
\begin{equation}\label{eq:total_dynamics_distinctp}
r(t+1)= G^T \Pi\, G\, (I_{MN}- \gamma \overline{L}(t) )\, r(t),
\end{equation}
where $\overline{L}(t)$ is defined in \eqref{eq:def_overlin_L}, and matrix $\Pi \in \mathbb{R}^{MN \times MN}$ is the block diagonal matrix given by
\[
\Pi=
\begin{bmatrix}
\overline{T}_1 & & & \\
& \overline{T}_2 & &\\
& & \ddots   &\\
&  &   & \overline{T}_N 
\end{bmatrix},
\]
with $\overline{T}_i = \Phi_i T_i \in \mathbb{R}^{M\times M}$.

Next, we denote $\overline{p} = \min_{i\in\mathcal V} {p_i}$, and define vector 
\begin{equation}\label{eq:definition_v_i}
v_i=\left((s_1)^{i-1}, (s_2)^{i-1}, \dots, (s_M)^{i-1}\right)\in \mathbb{R}^M,
\end{equation}
for each $i\in[\overline{p}+1]$, where $s_k$ is the communication key given in \eqref{eq:rik_encoding_function}.
Furthermore, denote 
\begin{equation}\label{eq:definition_overline_T_a}
\overline{T}_{a} =\Phi_{\ell} {T}_{\ell} \in \mathbb{R}^{M\times M},
\end{equation}
in which $\ell$ satisfies $p_{\ell}=\overline{p}$.
Then the next lemma reveals the invariance owned by the projection matrix $G^T \Pi G$.
\begin{lemma}\label{lem:property_GTG_Tp}
For any $p\in \mathbb{N}^N$, the following equalities hold
\begin{enumerate}
\item[i)] $G^T \Pi G \left(v_i \otimes \mathbf{1}_N  \right) = v_i  \otimes \mathbf{1}_N$, 
\item[ii)] $
G^T \Pi G \left(\overline{T}_{a} \otimes \mathbf{1}_N  \right) = \overline{T}_{a} \otimes \mathbf{1}_N$,
\end{enumerate} 
for all $i \in [\overline{p}+1]$.
\end{lemma}
\begin{proof}
Property i) follows the fact that 
$
\overline{T}_j v_i = v_i
$,
for any $j\in \mathcal{V}$, $i\in [\overline{p}+1]$.
Take $\ell\in\mathcal{V}$ such that $p_\ell=\overline{p}$, and we notice that the Vandermonde matrix defined in \eqref{eq:Vandermonde} satisfies $\Phi_{\ell}=(v_1, v_2, \dots, v_{\overline{p}+1})$. This leads to that $G^T \Pi G \left(\Phi_{\ell} \otimes \mathbf{1}_N  \right) = \Phi_{\ell}  \otimes \mathbf{1}_N$. Therefore, it holds that
\[
G^T \Pi G \left(\overline{T}_{a} \otimes \mathbf{1}_N  \right) = 
G^T \Pi G \left(\Phi_{\ell} \otimes \mathbf{1}_N  \right)
T_{\ell}
=
 \overline{T}_{a} \otimes \mathbf{1}_N,
\]
where we use the fact  $ \overline{T}_{a}=  \Phi_{\ell} T_{\ell}$ twice.
\end{proof}

Next, define  vector $\overline{r}_\infty = \frac{1}{N} \left( \overline{T}_{a}  \otimes \mathbf{1}_N^T \right) r(0) \in \mathbb{R}^M$. The following proposition states that the encoded state $r_i(t)\in \mathbb{R}^M$ of all the agents eventually achieve consensus as $t\to \infty$. 

\begin{prop}\label{prop:r_converge_diff}
Under the same assumptions with Theorem~\ref{thm:dif_degree}, it holds that
\[
\lim_{t\to \infty} r_i(t) =  \overline{r}_\infty, \quad a.s.,
\]
for each $i\in\mathcal V$.

\end{prop}

\begin{proof}
Define disagreement vector by 
\begin{equation}\label{eq:def_delta_diff_k}
\delta(t) = r(t) -  \overline{r}_\infty \otimes \mathbf{1}_N,
\end{equation}
and moreover we have the agreement goal $ \overline{r}_\infty \otimes \mathbf{1}_N $ satisfies
\begin{align}
 \overline{r}_\infty \otimes \mathbf{1}_N &=   \bigg( \frac{1}{N}\left( \overline{T}_{a}  \otimes \mathbf{1}_N^T \right) r(0) \bigg) \otimes \mathbf{1}_N \nonumber \\
&=   \frac{1}{N}\left( \overline{T}_{a}  \otimes \mathbf{1}_{N\times N} \right) r(0) \label{eq:pf_diff_decomposition_r_inf_1} \\
&= \frac{1}{N}\left( \overline{T}_{a}  \otimes \mathbf{1}_{N} \right) \left(I_{M}  \otimes \mathbf{1}_{N}^T \right) r(0).\label{eq:pf_diff_decomposition_r_inf_2}
\end{align}
Then the dynamics of disagreement $\delta(t)$ fulfills
\begin{align}\label{eq:dynamics_delta_diff} 
\delta(t+1) &= 
G^T \Pi G(I_{MN}- \gamma \overline{L}(t) )\, r(t) - \overline{r}_\infty \otimes \mathbf{1}_N \nonumber \\
&= G^T \Pi G \bigg( (I_{MN}- \gamma \overline{L}(t) )\, r(t) - \overline{r}_\infty \otimes \mathbf{1}_N \bigg) \nonumber \\
&= G^T \Pi G (I_{MN}- \gamma \overline{L}(t) ) \delta(t),
\end{align}
where the second equality is due to \eqref{eq:pf_diff_decomposition_r_inf_2} and ii) in Lemma~\ref{lem:property_GTG_Tp}, and the last equality follows \eqref{eq:prop_overline_L_2}.

Moreover, we can show that $(v_i \otimes \mathbf{1}_N)^T \delta(t)=0$, $\forall t$, where $v_i$ is defined in \eqref{eq:definition_v_i}  for $i\in[\overline{p}+1]$. This is because according to dynamics \eqref{eq:dynamics_delta_diff},
\begin{align*}
(v_i \otimes \mathbf{1}_N)^T \delta(t+1)
 &= (v_i \otimes \mathbf{1}_N)^T   G^T \Pi G (I_{MN}- \gamma \overline{L}(t) ) \delta(t)\\
 &=(v_i \otimes \mathbf{1}_N)^T \delta(t),
\end{align*}
where we use i) in Lemma~\ref{lem:property_GTG_Tp} with $i=1$, and the fact that $(v_i \otimes \mathbf{1}_N)^T \overline{L}(t)=0$, $\forall t$.
Then, this further implies that 
\begin{align*}
 (v_i \otimes \mathbf{1}_N)^T\delta(t) &= (v_i \otimes \mathbf{1}_N)^T \delta(0) \\
& = (v_i \otimes \mathbf{1}_N)^T \! \left( \!r(0) \!- \! \frac{1}{N}\left( \overline{T}_{a} \! \otimes \! \mathbf{1}_{N \times N}\right)  r(0) \!\right)\\
&=0,
\end{align*}
where we apply equation \eqref{eq:pf_diff_decomposition_r_inf_1} and   $\overline{T}_{a} v_i = v_i$. Therefore, we know the dynamics \eqref{eq:dynamics_delta_diff} indeed evolves in subspace $\overline{\mathcal S}$, where $\overline{\mathcal S}=\{z\in \mathbb{R}^{MN}: (v_i \otimes \mathbf{1}_N)^T  z=0, i\in [\overline{p}+1]\}$.

Next, we construct an orthogonal matrix $U=[U_1, U_2]\in \mathbb{R}^{M\times M}$ with $U_1=[v_1, v_2, \dots, v_{\overline{p}+1}]$ and $U_2 = [ u_1, u_2, \dots, u_{M-\overline{p}-1}]$. 
Then further denote subspace $\mathcal{S}_c\!=\!\Span\{u_1\otimes \mathbf{1}_N,  \dots, u_{M-\overline{p}-1} \otimes \mathbf{1}_N\}$, and it holds that
\begin{equation}\label{eq:decomposition_S_Sc}
\overline{\mathcal{S}} = \mathcal{S} \oplus \mathcal{S}_c,
\end{equation}
i.e., $\mathcal{S}_c$ is the orthogonal complement of subspace $\mathcal{S}$ in $\overline{\mathcal{S}}$, where subspace $\mathcal{S}$ is defined in \eqref{eq:subspace_S}.
In order to show \eqref{eq:decomposition_S_Sc}, we only need verify that 
\begin{enumerate}
\item[i)] $\mathcal{S} \subset \overline{\mathcal{S}}$ and $\mathcal{S}_c \subset \overline{\mathcal{S}}$,
\item[ii)] $\dim( \mathcal{S} + \mathcal{S}_c)=\dim(\overline{\mathcal{S}})$,
\item[iii)] $\mathcal{S} \perp \mathcal{S}_c$, 
\end{enumerate}
where i) follows the fact that $v_i^T u_j =0$, $\forall i\in[\overline{p}+1], j\in [M-\overline{p}-1]$, ii) is due to $\dim(\mathcal{S} + \mathcal{S}_c)= \dim(\mathcal{S})+\dim(\mathcal{S}_c) + \dim(\mathcal{S} \cap \mathcal{S}_c)=(MN-M)+(M-\overline{p}-1)+0=MN-\overline{p}-1$, and  iii) is because for any $z_s\in \mathcal{S}$, and $z_{s_c}\in \mathcal{S}_c$, we have $z_{s}^Tz_{s_c}=0$.

Now we can define the Lyapunov function for the dynamics \eqref{eq:dynamics_delta_diff} as 
 $V: \overline{\mathcal{S}} \to \mathbb R^+$ with $V(z)=z^Tz$.
Then for any $z\in \overline{\mathcal{S}}$ we have
\begin{align}\label{eq:pf_diff_E_V_1}
\,&\, \mathbb E\left[  V\bigg( G^T \Pi G (I_{MN}- \gamma \overline{L}(t) ) z  \bigg)\right]\nonumber\\
\leq &\,
\mathbb E\left[  \bigg\| \big( I_{MN} - \gamma \overline{L}(t)\big) z  \bigg\|^2 \right]\nonumber\\
=&\,
\mathbb E\left[  \bigg\| \big( I_{MN} - \gamma \overline{L}(t)\big) z_s +z_{s_c}  \bigg\|^2 \right]\nonumber
\\
=&\,
\mathbb E\left[  \bigg\| \big( I_{MN} \!-\! \gamma \overline{L}(t)\big) z_s\bigg\|^2  \!\!\! +\!\left\| z_{s_c} \right\|^2 \!\!\!+\!2 z_s^T\big( I_{MN}\! - \!\gamma \overline{L}(t)\big)z_{s_c}  \! \right]\nonumber
\\
\leq & \,
\overline{c}\left\| z_{s} \right\|^2 +\left\| z_{s_c} \right\|^2,
\end{align}
where $z= z_{s}+z_{s_c}$ with $z_s\in \mathcal{S}$, $z_{s_c}\in \mathcal{S}_c$ due to \eqref{eq:decomposition_S_Sc}, and $\overline{c}$ is defined in \eqref{eq:pf_E_V_2}. 
Note that in \eqref{eq:pf_diff_E_V_1} we use  $\sigma( G^T \Pi G) \subset \{0,1\}$ in the first inequality, apply $\overline{L}(t)z_{s_c}=0 $ to the first equality and the last inequality, and 
also apply a similar technique used in \eqref{eq:pf_E_V_2} to obtain the last inequality. 

In the end, as $\overline{c} \in [\frac{d_{max}-1}{2d_{max}-1},1)$, we have $\mathbb E\left[  V\left( G^T \Pi G (I_{MN}- \gamma \overline{L}(t) ) z  \right)\right] \leq \tilde{c} V(z)$ with $\tilde{c} < 1$, then the assertion follows by using Lemma~\ref{lem:invariance_sto_system}.
\end{proof}

Next, we are ready to prove Theorem~\ref{thm:dif_degree}.

\begin{proof}[\textbf{Proof of Theorem~\ref{thm:dif_degree}}]
By combining Lemma~\ref{lem:recon_unrecon_poly} and Proposition~\ref{prop:r_converge_diff}, Theorem~\ref{thm:dif_degree} follows with 
\begin{equation*}
\lim_{t\to \infty} x_i(t)=x_\infty,\qquad \forall i\in[N],
\end{equation*}
almost surely, for any initial condition $\{x_i(0)\}_{i\in \mathcal V}$, where 
\begin{equation}\label{eq:x_inf_different_degree}
x_\infty=\sum_{k\in [M]} [\overline{r}_\infty]_k  \frac{\prod_{\ell\in [M]\setminus \{k\}}  (-s_\ell)}{ \prod_{\ell\in [M] \setminus \{k\}}  (s_k-s_\ell)  }, 
\end{equation}
with $[\overline{r}_\infty]_k$ is the $k$-th element of $\overline{r}_\infty$.
\end{proof}

\begin{remark}\label{rmk:r_inf_different_degree}
For dynamics \eqref{eq:total_dynamics_distinctp}, we give in Appendix~\ref{sec:appendix_intuition} the intuition on how to find the steady state $\overline{r}_\infty = \frac{1}{N} \left( \overline{T}_{a}  \otimes \mathbf{1}_N^T \right) r(0)$ on which each node reaches consensus.
\end{remark}

\begin{remark}\label{rem:high_dim}
In Problem~\ref{prob:1}, if the agent's state has a dimension $d>1$, a natural way to extend the proposed method is to applying the  update law \eqref{eq:update_law_total} element-wise for each dimension. It will result in a consensus problem for in total $(M\times N \times d)$ encoded states. 
\end{remark}

\section{Numerical examples}\label{sec:simu}
To illustrate the theory, in this section, two numerical examples are given. In the first, we consider a network with all the nodes owning an identical privacy degree, and in the second, each node is allowed to have a distinct privacy degree.

\subsection{Example with common privacy degree}\label{subsec:identical_p}
We consider a network involving $N=6$ nodes, whose connection topology is specified by an undirected graph as shown in Figure~\ref{Fig:graph}.
To solve the problem of privacy-perserving consensus raised in Problem~\ref{prob:1} with all the nodes sharing a common  privacy degree $2$, i.e., $p= 2 \mathbf{1}_N$, we apply the proposed Algorithm~\ref{alg:handshake1} and update law \eqref{eq:update_law_total}. 
\begin{figure}[thp]
  \centering
  \includegraphics[width=0.25\textwidth]{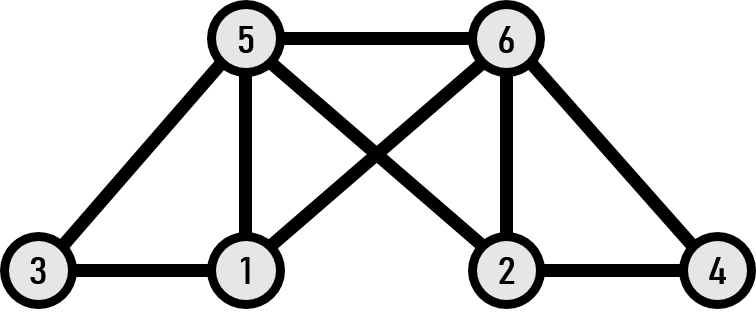}
  \caption{The interconnection graph of the network}\label{Fig:graph}
\end{figure}

According to the requirement in Theorem~\ref{thm:common_degree}, we take the number of channels $M=2d_{max}-1=7$, 
$\gamma = 0.95$, 
and communication keys $S=(1, 2, \dots, M)$. 
Moreover, for each node $i\in [N]$, the initial encoding functions $f_i(\theta,0)$ is constructed  following \eqref{eq:encoding_polynomial} with a set of randomized and privately known initial coefficients $\left\{a_{i1}(0),\, a_{i2}(0)\right\}$.
Then the initial encoded state can be accordingly computed as $r_i^{k}(0) = f_i(s_k,0)$, for $i\in\mathcal{V}, k\in[M]$. 

\begin{figure}[thp]
  \includegraphics[width=0.5\textwidth, trim={0.8cm 1cm 0 0cm},clip]{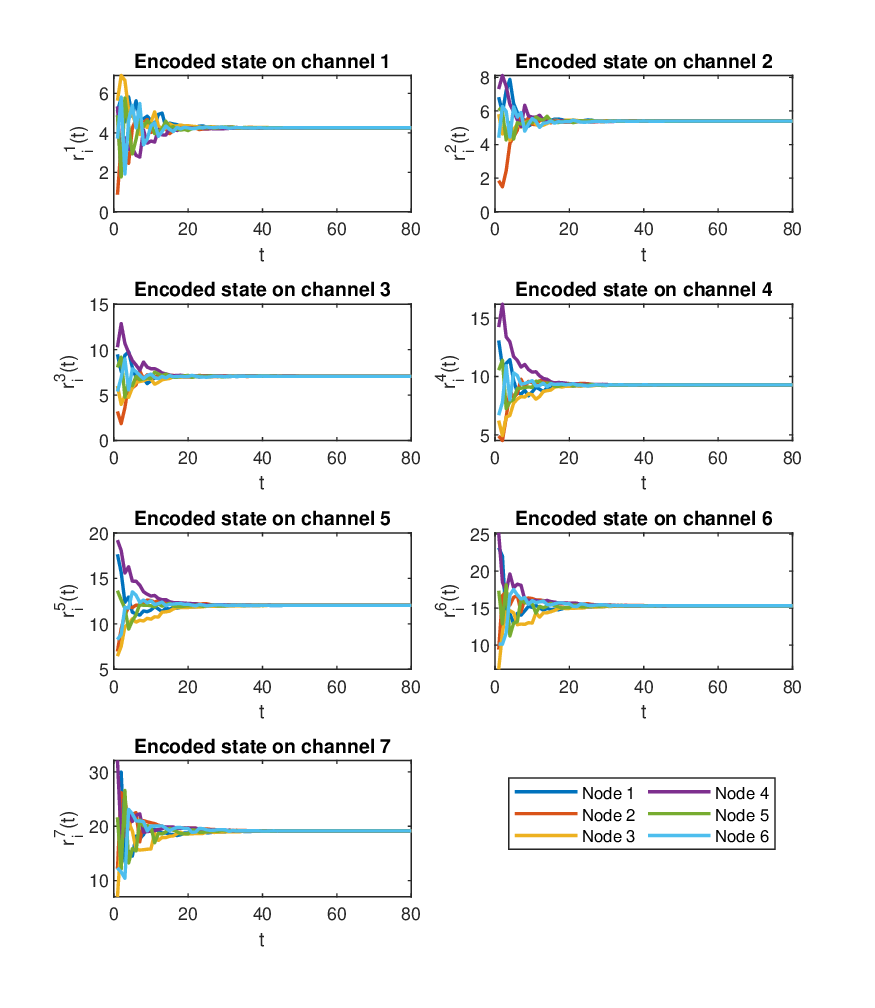}
  \caption{The trajectory of the encoded state on each channel $k\in[M]$ in the case with a common privacy degree.}\label{Fig:Estate_commonP}
\end{figure}

Next, we implement update law \eqref{eq:update_law_total} with the channel selection variable $c_{ij}(t)$ assigned through Algorithm~\ref{alg:handshake1} at the beginning of each time instant. 
The trajectory of the encoded state on all the channels are shown  in Figure~\ref{Fig:Estate_commonP}, and the state trajectory $x_i(t)$ is shown in Figure~\ref{Fig:state_commonP}-(a). As can be seen,  both  the encoded state corresponding to each channel and the state $x_i(t)$  reach consensus eventually for all the nodes.

Furthermore, we compare the convergence speed of the proposed privacy-preserving algorithm with that of the conventional consensus algorithm\footnote{
The conventional consensus algorithm under the undirected communication topology  is given by 
\begin{equation}\label{eq:converntional_consensus_alg}
x_i(t+1)= x_i(t) + \overline{\gamma}\sum_{j\in\mathcal V} \left( x_j(t)-x_i(t)\right),
\end{equation}
where the step size  $\overline{\gamma}\in(0,\frac{1}{d_{max}})$. More details can be found in \cite{olfati2007consensus}.  
}
which is shown in Figure~\ref{Fig:state_commonP}-(b). 
We see that the privacy-preserving method  compensates convergence speed for the secure of individual privacy, which is actually a natural result of the fact that the update law \eqref{eq:update_law_total} proceeds a consensus operation between pairs of neighbors only on one channel in each iteration.

\begin{figure}[t]
  \centering
  \includegraphics[width=0.5\textwidth, trim={0.2cm 0cm 0 0cm},clip]{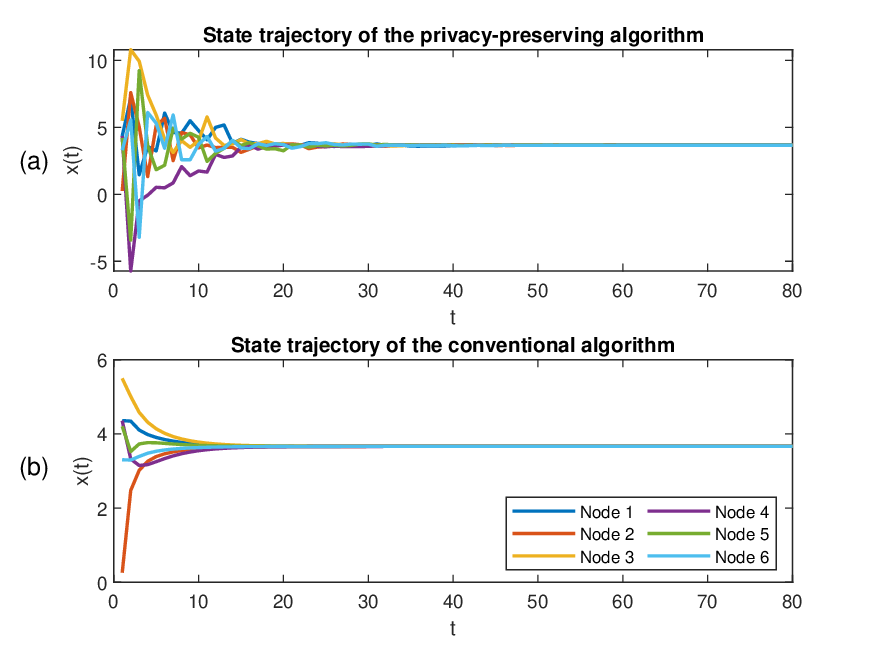}
  \caption{The state trajectory of the privacy-preserving and conventional consensus methods: (a) is the state trajectory given by the privacy-preserving method. (b) shows the sate trajectory of the conventional consensus algorithm \eqref{eq:converntional_consensus_alg} with $\overline{\gamma} = 1/(d_{max}+1)$.
  }\label{Fig:state_commonP}
\end{figure}

\subsection{Example with various privacy degrees}
We then consider the case of agents owing different privacy degree. The system consists of also $N=6$ nodes with a same network topology as that given in Figure~\ref{Fig:graph}. We set the privacy degree $p_i=|\mathcal{N}_i| - 1$, $i\in \mathcal V$, using which the privacy of a node $i$ is  disclosed only if all its neighbors collude, or all neighbors are attacked by an adversary.

The same number of channels $M$, step size $\gamma$, and communication keys $S$ are selected as those in Section~\ref{subsec:identical_p}.
Then the update law \eqref{eq:update_law_total} with the handshake procedure given in Algorithm~\ref{alg:handshake1} are simulated. 
The resulted trajectory of the encoded state on each channel $k\in[M]$ and trajectory of the state $x_i(t)$ are depicted in 
Figure~\ref{Fig:Estate_diffP} and Figure~\ref{Fig:state_diffP}, respectively.
We can see that the consensus is eventually reached for all the nodes.

\begin{figure}[thp]
  \centering
  \includegraphics[width=0.5\textwidth, trim={0.8cm 1cm 0 0cm},clip]{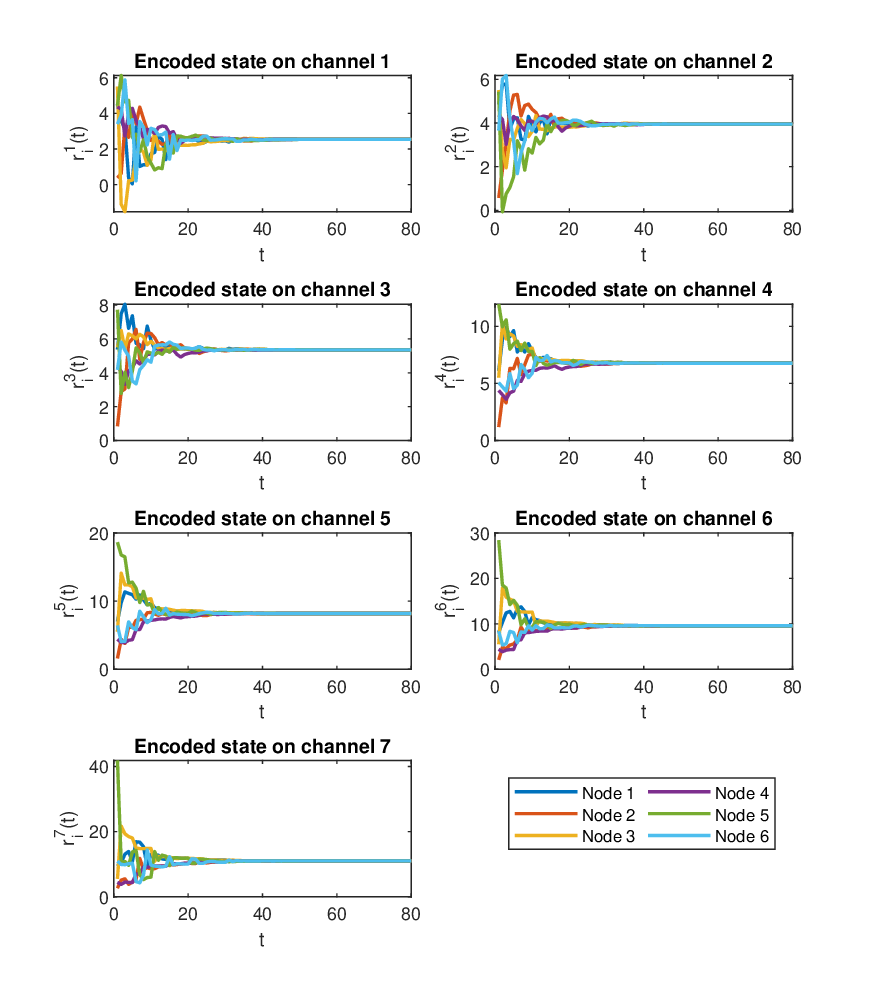}
  \caption{The trajectory of the encoded state on all channels in the case with different individual privacy degrees.}\label{Fig:Estate_diffP}
\end{figure}

\begin{figure}[t]
  \centering
  \includegraphics[width=0.45\textwidth, trim={0.2cm 0cm 0 0cm},clip]{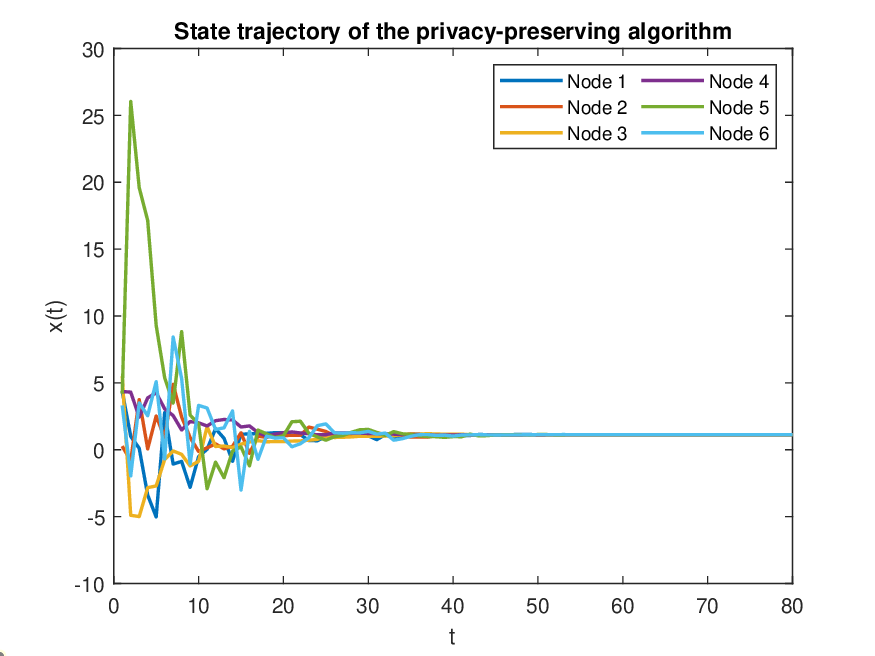}
  \caption{The state trajectory of the privacy-preserving consensus method in the case with different individual privacy degrees.
  }\label{Fig:state_diffP}
\end{figure}

\section{Conclusion and future work}\label{sec:conclusion}
In this work, we propose a privacy-preserving algorithm to solve the consensus problem based on the secret sharing schemes, in which a network of agents reach an agreement on their states without exposing their individual state until an agreement is reached.
Moreover, the proposed algorithm has an exact privacy degree. As a result, the system not only preserves the individual privacy, but also remains the reconstructability of individual states if more than a threshold amount of communication information is given. 

In  future work, such an idea can be extended to more problems involving network dynamics, such as formation control and distributed optimization. 
It would be also interesting to see following study to exploit the proposed algorithm in the relevant applications.
One application is privacy-preserving voting in networks, where a consensus problem is raised when each agents have an initial state of either $0$ or $1$ representing its preference for two candidates, and the majority can be found by averaging all the initial states. The other applications involve opinion agreement, sensor network averaging, survey mechanism, and distributed decision making.


\begin{appendix}
\subsection{Proof of Lemma~\ref{lem:invariance_sto_system}}\label{sec:appendix_proof}
\begin{proof}
First we denote $\mathcal{M}_t=\{x\in \mathcal S: V_t(x)=0\}$, then $\mathcal{M}=\liminf_{t\to \infty} \mathcal{M}_t$. Moreover, due to the radial undboundedness, function $V_t(x)$ is proper, i.e., for any compact set $A\subset\mathbb R$, we have $V_t^{-1}(A):=\{x\in \mathcal S: V_t(x) \in A\}$ is compact.  

According to Prop. 3.1 in \cite{zhang2016lasalle}, the solution $\{x(t)\}_{t\geq 0}$ satisfies
\[
\lim_{t\to \infty} V_t(x(t)) =0, \quad a.s.
\]
Then take the realization $(w^*_t)_{t\geq 0}$ such that the corresponding solution satisfies $\lim_{t\to \infty} V_t(x^*(t)) =0$, where 
\begin{align*}
x^*(t+1)&=F(x^*(t),w^*_t),\\
x^*(0)&=x(0).
\end{align*}
Note that the probability of $(w^*_t)_{t\geq 0}$ satisfying the above condition is exactly $1$. 

Due to the fact that $V_t(x)$ is proper, we have solution $\{x^*(t)\}$ is bounded. Then there is a positive limit set $L^+$ satisfying that for any $p\in L^+$, there is a subsequence $\{t_\ell\}_\ell \subset \mathbb N$ with $t_\ell \to \infty$ and $x^*(t_\ell)\to p$ as $\ell\to \infty$.
Further, as we know  $\lim_{t\to \infty} V_t(x^*(t)) =0$, it holds that  
$\lim_{t\to \infty} V_t(p) =0$. This implies that 
\[
L^+ \subset \liminf_{t\to \infty}\mathcal{M}_t,
\]
which proves the assertion.
\end{proof}

\subsection{Intuition to determine $\overline{r}_\infty$ in Section~\ref{sec:converge}.B}\label{sec:appendix_intuition}
Let $r_\infty \in \mathbb{R}^{MN}$ be the  point to which the dynamics \eqref{eq:total_dynamics_distinctp} almost surely converges. 
We first notice that due to property i) in Lemma~\ref{lem:property_GTG_Tp}, there are invariant variables in  dynamics \eqref{eq:total_dynamics_distinctp} being
\begin{align*}
(\widetilde{\Phi} \otimes \mathbf{1}_N)^T r(t+1) &= (\widetilde{\Phi} \otimes \mathbf{1}_N)^T (I_{MN}- \gamma \overline{L}(t) )r(t)\\
 &=(\widetilde{\Phi} \otimes \mathbf{1}_N)^T r(t) \\
 &=\dots = (\widetilde{\Phi} \otimes \mathbf{1}_N)^T r(0),
\end{align*} 
for any $t\in \mathbb N$, where $\widetilde{\Phi}=(v_1, v_2, \dots, v_{\overline{p}+1})\in \mathbb{R}^{M \times (\overline{p}+1)}$, and $v_i$ is defined in \eqref{eq:definition_v_i}. This further implies that 
\begin{equation}\label{eq:r_inf_relation}
(\widetilde{\Phi} \otimes \mathbf{1}_N)^T r_\infty = (\widetilde{\Phi} \otimes \mathbf{1}_N)^T r(0). 
\end{equation}
Next, we assume the following Assumption holds.
\begin{assumption}\label{ass:converge_to_consensus}
$r_\infty$ is a consensus point, i.e., $r_\infty=\overline{r}_\infty \otimes \mathbf{1}_N$, for some vector $\overline{r}_\infty \in \mathbb{R}^M$. 
\end{assumption} 

Moreover, as point $r_\infty$ is invariant under dynamics \eqref{eq:total_dynamics_distinctp}, i.e.,  $G^T \Pi\, G\, (I_{MN}- \gamma \overline{L}(t) )\, r_\infty =r_\infty$,  we have under Assumption~\ref{ass:converge_to_consensus},
\begin{align*}
G^T \Pi\, G\, \left(\overline{r}_\infty \otimes \mathbf{1}_N\right)  =\overline{r}_\infty \otimes \mathbf{1}_N.
\end{align*}
This further leads to that vector $\overline{r}_\infty$ satisfies 
\begin{equation}\label{eq:hat_r_inf_relation}
\overline{T}_i \overline{r}_\infty = \overline{r}_\infty, \quad \forall i\in \mathcal{V}.
\end{equation}
As $\overline{T}_i$ is the projection matrix on subspace $\Span  \Phi_i$, it holds that
$\overline{r}_\infty \in \Span \widetilde{\Phi}$.

Then, let $\overline{r}_\infty = \widetilde{\Phi} \vartheta$, for some $\vartheta \in \mathbb{R}^{\overline{p}+1}$. By using equation \eqref{eq:r_inf_relation}, we have
\[
(\widetilde{\Phi} \otimes \mathbf{1}_N)^T (\widetilde{\Phi} \otimes \mathbf{1}_N) \vartheta = (\widetilde{\Phi} \otimes \mathbf{1}_N)^T r(0). 
\]
Therefore, we can solve that $\vartheta= \frac{1}{N} (\widetilde{\Phi}^T \widetilde{\Phi}   )^{-1} (\widetilde{\Phi} \otimes \mathbf{1}_N)^T r(0)$. Finally, we have
\[
\overline{r}_\infty = \widetilde{\Phi} \vartheta=   \frac{1}{N}  (\overline{T}_a\otimes \mathbf{1}_N)^T r(0).
\]

In the end, as the convergence analysis in Section~\ref{sec:converge}.B is independent of any assumption, we conclude that Assumption~\ref{ass:converge_to_consensus} indeed holds.
\end{appendix}

\balance
\bibliographystyle{ieeetr}
\bibliography{ref}

\begin{IEEEbiography}
    {Silun Zhang} (S'16–M’20) received his B.Eng. and M.Sc. degrees in Automation from Harbin Institute of Technology, China, in 2011 and 2013, respectively, and the PhD degree in Optimization and Systems Theory from Department of Mathematics, KTH Royal Institute of Technology, Sweden, in 2019. 

Dr. Zhang is currently a Wallenberg postdoctoral fellow with the Laboratory for Information and Decision Systems (LIDS), MIT, USA. His main research interests include nonlinear control, networked systems, rigid-body attitude control, security and privacy in multi-party computation, and modeling large-scale systems.

\end{IEEEbiography}
\vspace{5mm}
\begin{IEEEbiography}
    {Thomas Ohlson Timoudas} received his PhD degree in Mathematics in October 2018 from KTH Royal Institute of Technology, Sweden, and his MSc (2013) and BSc (2012) degrees in Mathematics from Stockholm university, Sweden.

He is currently a postdoctoral researcher with the Department of Network and Systems Engineering at KTH Royal Institute of Technology, Sweden. His main research interests include dynamical systems, networked systems, internet of things, and distributed algorithms.
\end{IEEEbiography}
\begin{IEEEbiography}
{Munther A. Dahleh}(S’84–M’87–SM’97–F’01)
received his Ph.D. degree from Rice University, Houston, TX, in 1987 in Electrical and Computer Engineering. Since then, he has been with the Department of Electrical Engineering and Computer Science (EECS), MIT, Cambridge, MA, where he is now the William A. Coolidge Professor of EECS. He is also a faculty affiliate of the Sloan School of Management. He is the founding director of the newly formed MIT Institute for Data, Systems, and Society (IDSS). Previously, he held the positions of Associate Department Head of EECS, Acting Director of the Engineering Systems Division, and Acting Director of the Laboratory for Information and Decision Systems. He was a visiting Professor at the Department of Electrical Engineering, California Institute of Technology, Pasadena, CA, for the Spring of 1993. He has consulted for various national research laboratories and companies.

Dr. Dahleh is interested in Networked Systems with applications to Social and Economic Networks, financial networks, Transportation Networks, Neural Networks, and the Power Grid. Specifically, he focuses on the development of foundational theory necessary to understand, monitor, and control systemic risk in interconnected systems.  He is four-time recipient of the George Axelby outstanding paper award for best paper in IEEE Transactions on Automatic Control. He is also the recipient of the Donald P. Eckman award from the American Control Council in 1993 for the best control engineer under 35. He is a fellow of IEEE and IFAC. 
\end{IEEEbiography}

\end{document}